\documentclass[runningheads]{llncs}

\usepackage{graphicx}

\usepackage[english]{babel}
\usepackage{amsfonts}
\usepackage{inputenc}
\usepackage[ruled]{algorithm}
\usepackage{color}
\usepackage{algpseudocode}
\usepackage{booktabs,multicol,multirow}
\usepackage{fancyhdr}
\usepackage{epic}
\usepackage{xspace}
\usepackage{amssymb,amsmath,verbatim}
\usepackage{multirow}
\usepackage{multicol}
\usepackage{adjustbox}
\usepackage{threeparttable}
\usepackage{float}

\usepackage{amsthm}
\usepackage{chngcntr}
\usepackage{makecell}
\usepackage{tikz}
\usetikzlibrary{arrows}
\usepackage{lipsum}
\usepackage{enumitem}

\usepackage[misc,geometry]{ifsym}

\providecommand{\keywords}[1]{\textbf{\textit{Key words: }} #1}

\def\ZZ{\mathbb{Z}}

\def\RR{\mathbb{R}}

\def\cal{\mathcal}
\def\bf{\mathbf}
\def\calL{\mathcal{L}}

\def\TrapGen{\mathsf{TrapGen}}

\def\ExtBasis{\mathsf{ExtBasis}}

\def\Pr{\mathrm{Pr}}
\def\Adv{\mathsf{Adv}}

\def\u{\bf{u}}

\def\e{\bf{e}}

\def\x{\bf{x}}

\def\L{\Lambda}
\def\Lp{\Lambda^{\perp}}
\def\b{\bf{b}}
\def\s{\bf{s}}

\def\DuA{\cal{D}_{\L_q^{\u}(A),s}}

\tikzset{
  treenode/.style = {align=center, inner sep=0pt, text centered,
    font=\sffamily},
  arn_n/.style = {treenode, circle, white, font=\sffamily\bfseries, draw=black,
    fill=black, text width=1.5em},% arbre rouge noir, noeud noir
  arn_r/.style = {treenode, circle, red, draw=red, 
    text width=1.5em, very thick},% arbre rouge noir, noeud rouge
  arn_x/.style = {treenode, rectangle, draw=black,
    minimum width=0.5em, minimum height=0.5em}% arbre rouge noir, nil
}

\begin{document}
	\title{Lattice Blind Signatures with Forward Security }
%	\titlerunning{Lattice Blind Signatures with Forward Security}

	\authorrunning{H. Q. Le et al.}
	\author{Huy Quoc Le\inst{1,4}$^\textrm{(\Letter)}$ \and Dung Hoang Duong\inst{1}$^\textrm{(\Letter)}$\and Willy Susilo\inst{1}  \and Ha Thanh Nguyen Tran\inst{2} \and Viet Cuong Trinh\inst{3} \and Josef Pieprzyk \inst{4,5} \and  Thomas Plantard \inst{1} }
	\institute{Institute of Cybersecurity and Cryptology, School of Computing and Information Technology, University of Wollongong\\
		Northfields Avenue, Wollongong NSW 2522, Australia\\
		\email{qhl576@uowmail.edu.au, \{hduong,  wsusilo, thomaspl\}@uow.edu.au} 
		\and Department of Mathematical and Physical Sciences,
		Concordia University of Edmonton,
		7128 Ada Blvd NW, Edmonton,
		AB T5B 4E4, Canada \\
		\email{hatran1104@gmail.com}
		\and 
		Faculty of Information and Communication Technology, Hong Duc University,\\ 565 Quang Trung, Thanh Hoa, Vietnam\\
		   \email{trinhvietcuong@hdu.edu.vn}
		\and 	CSIRO Data61, Sydney, NSW, Australia.\\
		\and 
		Institute of Computer Science, Polish Academy of Sciences, Warsaw, Poland.\\
		\email{Josef.Pieprzyk@data61.csiro.au}
	}

	\maketitle              
	
\begin{abstract} Blind signatures play an important role in both electronic cash and electronic voting systems. 
Blind signatures should be secure against various attacks (such as signature forgeries). 
The work puts a special attention to secret key exposure attacks, which totally break digital signatures.
Signatures that resist secret key exposure attacks are called forward secure in the sense
that disclosure of a current secret key does not compromise past secret keys. 
This means that forward-secure signatures must include 
a mechanism for secret-key evolution over time periods.

This paper gives a construction of the \textit{first} blind signature that is forward secure. 
The construction is based on the SIS assumption in the lattice setting. 
The core techniques applied are the binary tree data structure 
for the time periods and the trapdoor delegation for the key-evolution mechanism.  	
	
	\end{abstract}
	%
	%
	%%%%%%%%%%%%%%%%%%%%%%%%%%%%%%
\section{Introduction}

Key exposure is one of most serious dangers for both secret and public key cryptography.
When secret keys are disclosed, cryptographic systems using them are completely broken.
Fortunately, there are some solutions that can be used to mitigate secret-key exposure.
They are summarized in \cite{BM99}.
Among many possible solutions, forward security seems to be the most promising when trying
to minimize a damage caused by secret-key disclosure.

For cryptographic protocols, 
forward security guarantees that even if the current session key is compromised by an adversary, 
she gets no information about previous session keys. 
This means that past sessions are still secure.
The notion of \textit{forward security} has been coined by G\"{u}nther in \cite{Gun90} and later used in  \cite{DOW92}
to evaluate security of authenticated key-exchange protocols.
Note that the authors of \cite{DOW92} and \cite{Gun90} call it \textit{forward secrecy}.
Ross Anderson in \cite{And02} extends the notion for digital signatures. 

Blind signatures, introduced by Chaum \cite{Chau83}, 
allow users to obtain message signatures from a signer without leaking information about message contents.
Blind signatures are indispensable in many applications such as 
electronic cash \cite[Section 1]{PS96} and electronic voting protocols \cite{Kuc10}.  
For such security-critical applications, one would expect blind signatures to be resistant against key disclosure.
An obvious solution is to incorporate forward security into blind signatures.
There are many works such as \cite{CHYC05,DCK03,JFC+10} that follow this line of investigation.
All solutions published so far rely on number-theoretic assumptions and consequently
are insecure against quantum adversaries.\\

\noindent \textbf{Related Works.} 
Bellare and Milner investigate secret-key exposure of digital signatures in their Crypto99 paper \cite{BM99}.
They formulate a security model and define forward-secure digital signatures.
They also design their forward-secure signature assuming intractability of integer factorization. 
Abdalla and Reyzin \cite{AR00}, and Itkis and Reyzin \cite{IR01} improve efficiency
the Bellare-Miner signature. 
The work of Duc et al. \cite{DCK03} is the first, which investigates forward security in the context of blind signatures.
The authors of  \cite{DCK03} adopt the definition and security model 
from \cite{BM99} to forward-secure blind signatures.
%noting that the blindness security is  not necessarily considered in two different time points, 
%while the forward security is sticked to the unforgeability (forward-secure unforgeability), not one-more unforgeability. 
Their blind signature provides forward-secure unforgeability assuming intractability of the strong RSA problem
and access to random oracle.
Their security proof exploits the forking lemma by Pointcheval and Stern \cite{PS96}. 
Later, Chow et al. \cite{CHYC05} design forward-secure blind signature using bilinear pairings.
Jia et al. \cite{JFC+10} describe a forward-secure blind signature that is also based on bilinear pairings.
Boyd and Gellert \cite{BG19} give a comprehensive survey of methods of incorporating 
forward security to different cryptographic primitives. 
They also unify different approaches to forward security by generalising the notion and its terminology.\\

\noindent \textbf{Our Contributions and Approach.} 
Thanks to its quantum resistance,  lattice-based cryptography is attracting more and more 
attention from the research community.
However, there is no lattice-based construction of forward-secure blind signatures.
Our work fills the gap. We construct \textit{the first forward-secure blind signature in the lattice setting}. 
Forward security is proven in the random oracle model assuming intractability of the average case of short integer 
solution (SIS). We also use the rewinding (forking lemma) argument. 

Inspired by the works \cite{LDS20,Ruc10,ZJZ+18}, 
our signature is designed using the 3-move Fiat-Shamir transformation.
To achieve blindness, the rejection sampling technique is applied (see Section \ref{lattices}).
Thus, an extra move is needed to ensure that a final signature is valid.
In order to achieve forward security, we exploit both
a binary tree structure for lattice-based schemes introduced in \cite{CHKP10} 
and a trapdoor delegation from \cite{ABB10-EuroCrypt,CHKP10}. 

To obtain forward-secure signature, we need a mechanism that permits for a secret-key update
between two time intervals.
For this purpose, we use a binary tree of the depth $\ell$, 
whose leaves are labelled from left to right by consecutive time intervals $t=0$ up to $t=\tau-1$,
where $\tau=2^{\ell}$ is the total number of time intervals.
To generate the public key and the initial secret key, we choose random matrices 
$A_j^{(0)}$, $A_j^{(1)}$ for $j \in [\ell] $ together with a matrix/trapdoor pair  ($A_0, T_{A_0}$). 
Now, for any node $w^{(i)}=(w_1,\cdots, w_{i}) \in \{0,1\}^i$, 
we build up a concatenated matrix of form $F_{w^{(i)}}=[A_0\|A_1^{(w_1)} \|A_i^{(w_i)}]$. 
Then, we can compute a trapdoor for $\L_q^{\bot}(F_{w^{(i)}})$ using $T_{A_0}$. 
If the node $w^{(k)}$ is the ancestor of the node $w^{(i)}$, then we can obtain a trapdoor for $\L_q^{\bot}(F_{w^{(i)}} )$ 
from a trapdoor for $\L_q^{\bot}(F_{w^{(k)}} )$. 
However, one cannot get a trapdoor for $\L_q^{\bot}(F_{w^{(k)}} )$ from a trapdoor of $\L_q^{\bot}(F_{w^{(i)}} )$. 
This is the main idea behind the key evolution (key update) mechanism.

%---section-Lattices-------
\section{Preliminaries} \label{lattices}
For a positive integer $\ell$, $[\ell]$ stands for the set $\{1, \cdots, \ell\}$. 
For a vector $\bf{c}$ and a matrix $S$, 
$\bf{c}[i]$ and $S[i]$ represent the $i$-th element of $\bf{c}$ and the $i$-th column of  $S$, respectively.\\

\noindent \textbf{Lattices}. Integer lattices are discrete subgroups of $\ZZ^m$. Formally, a lattice $\calL$ in $\ZZ^m$ is defined as
$\calL=\calL(B):=\left\{\sum_{i=1}^n\b_ix_i : x_i\in\ZZ,\forall i=1,\cdots,n \right\}\subseteq\ZZ^m,$
where $B=[\b_1,\cdots,\b_n]\in\ZZ^{m\times n}$ is called a basis of $\calL$, and $\b_i$'s are column vectors.
We call $n$ the rank of $\calL$. We say $\calL$ is a full rank lattice  if $n=m$. 

Given a matrix $A\in\ZZ^{n\times m}$ and a vector  $\bf{u}\in\ZZ_q^n$, we define two lattices: 
\begin{align*}
%\L_q(A) &:= \left\{ \e\in\ZZ^m ~\rm{s.t.}~ \exists \bf{s}\in\ZZ_q^n~\rm{where}~A^T\bf{s}=\bf{e}\mod q \right\},\\
\Lp_q(A) &:= \left\{ \e\in\ZZ^m~\rm{s.t.}~A\e={\bf 0}\mod q \right\}, \\
\L_q^{\bf{u}}(A) &:=  \left\{ \e\in\ZZ^m~\rm{s.t.}~A\e=\bf{u}\mod q \right\}.
\end{align*}
They all are full rank lattices containing $q\ZZ^m$ and are called $q$-ary lattices. 
Note that if $\bf{v}\in\L_q^{\bf{u}}(A)$, then $\L_q^{\bf{u}}(A)=\Lp_q(A)+\bf{v}$.

For a set of vectors $S=\{\s_1,\cdots,\s_k\}$ in $\mathbb{R}^m$,  
we denote $\|S\|:=\max_i\|\s_i\|$. 
Also,  $\widetilde{S}:=\{\widetilde{\s}_1,\cdots,\widetilde{\s}_k \}$ stands for the 
Gram-Schmidt orthogonalization of the vectors $\s_1,\cdots,\s_k$ in that order. 
The Gram-Schmidt norm of $S$ is denoted by $\|\widetilde{S}\|$. 
A basis of a lattice is called \textit{short} if its Gram-Schmidt norm is short.

We recall the shortest independent vectors problem (SIVP), 
which is the worst case of approximation problem on lattices. 
Note that the \textit{$i$-th minimum} of a $n$-dimensional lattice $\mathcal{L}$ is defined as      
$\lambda_i(\mathcal{L}):=\min\{r: \dim(\text{span}(\mathcal{L} \cap \mathcal{B}_n(0,r))) \geq i\}$, 
where $\mathcal{B}_n(0,r)=\{\mathbf{x} \in \mathbb{R}^n: \Vert \mathbf{x}\Vert \leq r \}$.   

\begin{definition}[SIVP] \label{sivp}
Given a full-rank basis $B$ of an $n$-dimensional lattice $\calL$.
$\textsf{SIVP}_{\gamma}$ requires to output a set of $n$ linearly independent lattice 
vectors $S \subset \calL(B)$ such that $\Vert S\Vert \leq \gamma(n)\cdot \lambda_n(\calL(B))$.
\end{definition}

\noindent 
Below we define 
discrete Gaussian distribution over an integer lattice.
\begin{definition}[Gaussian Distribution]
	Let $\L\subseteq\ZZ^m$ be a lattice. For a vector $\mathbf{v}\in\RR^m$ and a positive parameter $s\in\RR$, define
	$\rho_{s,\mathbf{v}}(\x)=\exp\left(-\frac{\pi\|\x-\mathbf{v}\|^2}{s^2}\right)$ and $ 	\rho_{s,\bf{v}}(\L)=\sum_{\x\in\L}\rho_{s,\mathbf{v}}(\x).    $
	The discrete Gaussian distribution over $\L$ with center $\mathbf{v}$ and parameter $\sigma$ is
	$\forall \bf{y}\in\L,\cal{D}_{\L,s,\mathbf{v}}(\bf{y})=\frac{\rho_{s,\mathbf{v}}(\bf{y})}{\rho_{s,\mathbf{v}}(\L)}.$
\end{definition}
For convenience,  $\rho_s$ and $\cal{D}_{\L.s}$ denote $\rho_{\bf{0},s}$ and $\cal{D}_{\L,s,\bf{0}}$, respectively. 
When $s=1$, we will write $\rho$ instead of $\rho_1$. 
Also, $\cal{D}^m_{s,\mathbf{v}}$ and $\cal{D}^m_{s}$ stand for  
$\mathcal{D}_{\mathbb{Z}^m,s,\mathbf{v}}$ and $\mathcal{D}_{\mathbb{Z}^m,s}$, respectively.
%%%%% END
\begin{lemma}[ {\cite[Lemma 4.5]{Lyu12}}] \label{lem5} 
	For any $\mathbf{v} \in \mathbb{Z}^m$, if $s=\alpha \cdot \Vert \mathbf{v}\Vert$, where $\alpha>0$, we have
	$\Pr\left[ {\mathcal{D}_{s}^m(\mathbf{x})}/{\cal{D}^m_{s,\mathbf{v}}(\mathbf{x})}\leq e^{12/\alpha+1/(2\alpha^2)}: \mathbf{x}\leftarrow \mathcal{D}_s^m\right] \geq 1-2^{-100}.$
\end{lemma}
\begin{remark}\label{rem2}
	In Lemma \ref{lem5}, if $\alpha=12$, i.e., $s=12\Vert \mathbf{v} \Vert $ then  ${\mathcal{D}_{s}^m(\mathbf{x})}/{\mathcal{D}_{s,\mathbf{v}}^m(\mathbf{x})}$ $\leq e^{1+1/288}$ with probability not smaller than $1-2^{-100}$. 
\end{remark}       
	
\noindent \textbf{Trapdoors and Trapdoor Delegation.}	
Alwen and Peikert~\cite{AP09} give an algorithm for sampling 
a uniform matrix $A\in\ZZ_q^{n\times m}$ together with a  short basis $T_A$ for $\Lp_q(A)$.
It is an improvement of the algorithm published by Ajtai in~\cite{Ajtai99}.
We call $T_A$ an  \textit{associated trapdoor} for $A$ or for $\Lp_q(A)$. 
	\begin{theorem}[\cite{AP09}]\label{thm:TrapGen}
		Let $q\geq 3$ be odd and $m:=\lceil 6n\log q\rceil$. There is a probabilistic polynomial-time (PPT) algorithm $\TrapGen(q,n)$ that outputs a pair $(A\in\ZZ_q^{n\times m},T_A \in\ZZ^{m\times m})$ such that $A$ is statistically close to a uniform matrix in $\ZZ_q^{n\times m}$ and $T_A$ is a basis for $\Lp_q(A)$ satisfying
		$\|\widetilde{T_A}\|\leq O(\sqrt{n\log q})\text{ and }\|T_A\|\leq O(n\log q)$
		with all but negligible probability in $n$.
	\end{theorem}
\noindent
Regarding Gaussian distribution, $q$-ary lattices and trapdoors, 
some useful results are presented in  the following lemma and theorem. 
	\begin{lemma}[{\cite[Corollary 5.4]{GPV08}}] \label{uniform}
	 Let $m, n, q$ be positive integers such that $q$ is prime and $m \geq 2n \log q$ . Then for all but $2q^{-n}$ fraction of all matrix $A\in \ZZ_q^{n \times m}$ and for any $s\geq \omega(\sqrt{\log m})$, the distribution of $\bf{u}:=A\bf{e} \text{ (mod } q)$ is statistically close to uniform over $\ZZ_q^n$, where $\bf{e} \leftarrow \cal{D}_{\ZZ^m,s}$. Furthermore, the conditional distribution of $\bf{e} \leftarrow \cal{D}_{\ZZ^m,s}$, given $A\bf{e}=\bf{u} \text{ (mod } q)$, is exactly $\cal{D}_{\L_q^{\textbf{u}}(A),s}$.
	\end{lemma}
		
\begin{theorem}\label{thm:Gauss}
	Let $q> 2$ and let $A, B$ be a matrix in $\ZZ_q^{n\times m}$ with $m>n$. 
	Let $T_A, T_B$ be a basis for $\Lp_q(A)$ and  $\Lp_q(B)$, respectively. 
	Then the following statements are true.
		\begin{enumerate}
%%\item { \cite[Lemma 4.4]{Lyu12}} \label{lem2}  For any $\eta>0$, we have $\Pr[ \mathbf{z} \leftarrow \mathcal{D}_s^m:\Vert \mathbf{z} \Vert >\eta s\sqrt{m}] \leq \eta^m e^{\frac{m}{2}(1-\eta^2)}.$
	
		    \item {\cite[Lemma 4.4]{MR04}} For $s\geq\|\widetilde{T_A}\|\cdot  \omega(\sqrt{\log n})$, we have $$\Pr[\x\gets\DuA~:~\|\x\|>s\sqrt{m}]\leq\mathsf{negl}(n).$$
	%	\item {\cite[Theorem 17]{ABB10-EuroCrypt}}  Let $M$ be a matrix in $\ZZ_q^{n\times m_1}$ and $s\geq\|\widetilde{T_A}\|\cdot \omega(\sqrt{\log(m+m_1)})$. Then there exists a PPT algorithm $\SampleLeft(A,M,$ $T_A,\mathbf{u},s)$ that outputs a vector $\e\in\ZZ^{m+m_1}$ distributed statistically close to $\cal{D}_{\L_q^{\u}(F_1),s}$ where $F_1:=(A~\|~M)$. In particular $\e\in \L_q^{\mathbf{u}}(F_1)$, i.e., $F_1\cdot\e=\mathbf{u}\mod q$.
	%	\item  {\cite[Theorem 19]{ABB10-EuroCrypt}} Let $R$ be a matrix in $\ZZ^{k\times m}$ and let $s_R:=\sup_{\|\x\|=1}\|R\x\|$. Let $F_2:=(A~\|~AR+B)$. Then for  $s\geq\|\widetilde{T_B}\|\cdot s_R\cdot \omega(\sqrt{\log m})$, there exists a PPT algorithm $\SampleRight(A,B,R,T_B,\mathbf{u},s)$ that outputs a vector $\e\in\ZZ^{m+k}$ distributed statistically close to $\cal{D}_{\L_q^{\mathbf{u}}(F_2),\s}$. In particular $\e\in \L_q^{\u}(F_2)$, i.e., $F_2\cdot\e=\mathbf{u}\mod q$. 

	\item \cite[Theorem 4.1]{GPV08} There is a PPT algorithm $\mathsf{SampleD}(B, s, \mathbf{v}$) that, given a basis $B$ of an $n$-dimensional lattice $\L:=\calL(B)$, a parameters $s \geq \Vert \widetilde{B} \Vert \cdot \omega(\sqrt{\log n})$ and a center $\mathbf{v} \in \mathbb{R}^n$, outputs a sample from a distribution statistically close to  $\cal{D}_{\L,s,\mathbf{v}}$.

	\item {\cite[Subsection 5.3.2]{GPV08}} There is a PPT algorithm $\mathsf{SampleISIS}(A,T_A,s,\mathbf{u})$ that, on input a matrix $A$, its associated trapdoor $T_A$, a Gaussian parameter $s \geq \Vert \widetilde{T_A} \Vert \cdot \omega(\sqrt{\log n})$ and a given vector $\mathbf{u}$,  outputs a vector $\mathbf{e}$ from $ \cal{D}_{\L_q^{\textbf{u}}(A),s}$. It performs as follows: first it chooses an arbitrary $\mathbf{t}\in \ZZ^m$ satisfying that $A\bf{t}=\bf{u} \text{ (mod } q)$ ($\bf{t}$ exists for all but an at most $q^{-n}$ fraction of $A$).  It then samples $\mathbf{w} \leftarrow \cal{D}_{\L_q^{\bot}(A),s}$ using $\mathsf{SampleD}( T_A,s,-\mathbf{t})$ and finally outputs $\mathbf{e}=\mathbf{t}+\mathbf{w}.$
	\item {\cite[Section 2]{LDS20}}
	There is a PPT  algorithm $\mathsf{SampleKey}(A, T_A, s, K$) that takes as input a matrix $A \in \mathbb{Z}_q^{n \times m}$, its associated trapdoor $T_A\in \mathbb{Z}_q^{m \times m}$, a real number $s \geq \Vert \widetilde{T_A} \Vert \cdot \omega(\sqrt{\log n})$ and  matrix $ K \in \mathbb{Z}_q^{n \times k}$ to output a random (column) matrix $S\in \mathbb{Z}^{m \times k}$ such that the $j$-th column $S[j] \in \mathsf{Dom}:= \{\mathbf{e} \in \mathbb{Z}^m: \Vert \mathbf{e} \Vert \leq s\sqrt{m} \}$ for all $j\in[k]$ and that $A \cdot S=K \text{ (mod } q)$ with overwhelming probability. The distribution of $S$  is $\mathcal{D}_{\mathbb{Z}^{m\times k},s}$ statistically close to the uniform distribution over $\mathsf{Dom}^k$. It performs by calling $k$ times the algorithm $\mathsf{SampleISIS}(A,T_A,s,\mathbf{u})$ in which $\mathbf{u}=K[j]$ for  $j \in \{1, \cdots, k\}$.

		%	\item {\cite[Lemma 3.2]{CHKP10}} There is a deterministic polynomial-time algorithm $\mathsf{ExtBasis}(A,F=A \|B, T_A)$ that works as follows: On input an arbitrary matrix $A \in \mathbb{Z}_q^{n\times m_1}$ whose columns generate the entire group  $\mathbb{Z}_q^{n}$, an arbitrary matrix $B \in \mathbb{Z}_q^{n\times m_2}$, ($m_1,m_2$ may be arbitrary) and $T_A$ is a basis of $\Lambda_q^{\bot}(A)$, it outputs a basis $T_F$ of $\Lambda_q^{\bot}(F)$ with $F=A\|B \in \mathbb{Z}_q^{n\times (m_1+m_2)}$ such that $\Vert \widetilde{T_F}\Vert=\Vert \widetilde{T_A}\Vert.$
	\end{enumerate}
\end{theorem}

\noindent
In order to securely delegate a basis for an extended lattice, 
one can call the $\mathsf{ExtBasis}$ algorithm described below.
\begin{lemma}[{\cite[Theorem 5]{ABB10-EuroCrypt}}] Let $A:=[A_1\|A_2\|A_3]$ be a concatenation of three matrices 
$A_1,A_2,A_3$. Suppose that $T_{A_2}$ is a basis of $\Lp_q(A_2)$. 
Then, there is a deterministic polynomial time algorithm  $\mathsf{ExtBasis}(A,T_{A_2})$ 
that outputs a basis $T_A$ for $\Lp_q(A)$ such that $\|\widetilde{T_A}\|=\|\widetilde{T}_{A_2}\|$.
\end{lemma}

\noindent \textbf{Hardness Assumption}. 
Forward-security  of our construction is proven assuming hardness of the SIS problem.
\begin{definition}[$l_2$-\textsf{SIS}$_{q,n,m,\beta}$ problem, {\cite[Definition 3.1]{Lyu12}}] \label{def2}
	Given a random matrix $A \leftarrow_{\$} \mathbb{Z}_q^{n \times m}$, find a vector $\mathbf{z}\in \mathbb{Z}^{m} \setminus \{\mathbf{0}\}$ such that $A\mathbf{z} =\mathbf{0} \text{ (mod } q)$ and $\Vert \mathbf{z}\Vert \leq \beta.$  
\end{definition}

\noindent
The hardness of $l_2$-\textsf{SIS} is stated by the following theorem.

\begin{theorem}[{\cite[Proposition 5.7]{GPV08}}]\label{thm:SIS}%\cite{Regev05}
For any poly-bounded $m$, $\beta=poly(n)$ and for any prime $q \geq \beta \cdot \omega(\sqrt{n\log n})$, 
the average case problem $l_2-\mathsf{SIS}_{q,n,m,\beta}$ is as hard as approximating the SIVP problem 
(among others) in the worst case for a factor $\gamma=\beta\cdot \tilde{O}(\sqrt{n})$.
\end{theorem}
\noindent
Define the \textsf{SIS}$_{q,n,m,d}$ distribution  by the pair $(A, A\bf{s})$,
where $A \xleftarrow{\$} \mathbb{Z}_q^{n \times m}$ and 
$\bf{s} \xleftarrow{\$} \{-d, \cdots, 0, \cdots, d\}$ are chosen at random.
The distribution is characterised by the following lemma.
\begin{lemma}[Discussed in \cite{Lyu12}] \label{lem8}
	For $d \gg q^{m/n}$, the $\mathsf{SIS}_{q,n,m,d}$ distribution is statistically close to uniform
	over $\mathbb{Z}_q^{n \times m} \times \mathbb{Z}_q^{n}$.
	Given $(A, \bf{u})$ from the $\mathsf{SIS}_{q,n,m,d}$  distribution, there are many possible solutions $\bf{s}$ satisfying
	$A\bf{s} = \bf{u}$.
\end{lemma}
\noindent \textbf{Rejection Sampling}.  
This is an aborting technique that is frequently used in lattice-based cryptography. 
The technique plays an important role in guaranteeing the blindness 
as well as it is used in simulation of the forward-security proof for our signature. 
\begin{lemma}[Rejection Sampling, {\cite[Theorem 4.6]{Lyu12}}]
	\label{lem3}
	Let $V=\{\mathbf{v}\in \mathbb{Z}^m: \Vert \mathbf{v} \Vert \leq \delta \}$ be a subset of $\mathbb{Z}^m$ 
	and  $s=\omega(\delta\log\sqrt{m})$ be a real number. Define a probability distribution $h:V \rightarrow \mathbb{R}$. 
	Then there exists a universal $M=O(1)$ satisfying that two algorithms $\mathcal{A}$ and $\mathcal{B}$ defined as:
	\begin{enumerate}
	\item ($\mathcal{A}$): $\mathbf{v} \leftarrow h$, $\mathbf{z} \leftarrow \mathcal{D}_{\mathbf{v},s} ^m$, 
		output $(\mathbf{z}, \mathbf{v} )$ with probability 
		$\min(\frac{\mathcal{D}_{s} ^m(\mathbf{z},)}{M\mathcal{D}_{\mathbf{v},s} ^m(\mathbf{z})},1)$, and
	\item ($\mathcal{B}$): $\mathbf{v} \leftarrow h$, $\mathbf{z} \leftarrow \mathcal{D}_{s}^m $, 
		output $(\mathbf{z}, \mathbf{v} )$ with probability $1/M$,
	\end{enumerate}
	have a negligible statistical distance $\Delta(\mathcal{A}, \mathcal{B}):=2^{-\omega(\log m)}/M$.
	Moreover, the probability that $\mathcal{A}$ outputs something is at least $(1-2^{-\omega(\log m)})/M$. 
	In particular, if $s=\alpha \delta$ for any $\alpha >0$, then $M=e^{12/\alpha+1/(2\alpha^2)}$,  
	$\Delta(\mathcal{A}, \mathcal{B})=2^{-100}/M$ and the probability that $\mathcal{A}$ 
	outputs something is at least $(1-2^{-100})/M$.
\end{lemma}
\noindent \textbf{Commitment Functions}. A commitment function
$\mathsf{com}$ maps a pair of two strings ($\mu$,$\mathbf{d}) \in \{0,1\}^* \times \{0,1\}^n$ 
(called \textit{committed string}) to a \textit{commitment string}  $C:=\mathsf{com}(\mu, \mathbf{d}) \in\{0,1\}^n$. 
We need \textsf{com} that is both 
 \textit{statistically hiding} and \textit{computationally binding}. 
 For more details, see \cite{Ruc10}, \cite{LDS20}.

%----section-forward-secure-blind-sig-----
\section{Framework of Forward-secure Blind Signatures}
In this section, we recap the syntax and the security model for
forward-secure blind signatures (\textsf{FSBS}). 
We follow \cite{DCK03}, which is in turn adapted from \cite{BM99}.
\subsection{Syntax of Forward-secure Blind Signature Schemes} \label{brs}
A forward-secure blind signature (or $\textsf{FSBS}$ for short) consists of 
the four algorithms \textsf{Setup}, \textsf{KeyUp}, \textsf{Sign}, and \textsf{Verify}. 
They are described as follows:
\begin{itemize}
	\item $(pp, pk,sk_0) \xleftarrow{\$} \textsf{Setup}(1^{n})$. 
	The algorithm is a PPT one that takes as input a security 
	parameter $n$ and 
	generates common parameters $pp$, a public key $pk$  and an initial secret key $sk_\epsilon$. 
	\item $sk_{t+1} \xleftarrow{\$} \textsf{KeyUp}(sk_{t}, t$): 
	The key update algorithm is a PPT one, which derives
	a secret key $sk_{t+1}$ for the time period $t+1$ from a secret key 
	$sk_{t}$ for a time period $t$. 
	After execution,
	the algorithm deletes the secret key  $sk_t$.
	
	\item $( \mathcal{V}, \Sigma) \xleftarrow{\$} \textsf{Sign}(pp, pk, sk_t,t,\mu)$:  
	The signing algorithm involves an interaction between a user, say $\mathcal{U}(pp, pk, t, \mu)$ 
	and a signer, say $\mathcal{S}(pp, pk, sk_t, t)$. 
	At a time period $t$, the user  blinds the message $\mu$ using the secret key $sk_t$ and sends
	it to the signer. The signer replies with a signature of the blinded message.
	After successful interactions, the user obtains a signature $\Sigma$ of the original message $\mu$
	at the time $t$. The signer gets its own view $\mathcal{V}$. 
	 If the interaction fails, the user and signer output $\Sigma:=\bot$ and $\mathcal{V}:=\bot$, respectively.
	\item $1/0 := \textsf{Verify}(pp, pk, t, \mu, \Sigma)$: 
	The verification algorithm is a deterministic one that outputs either $1$ if $\Sigma$ is non-$\bot$ and valid 
	or $0$, otherwise. As the input, it accepts a parameter $pp$, a public key $pk$, a time period $t$, 
	a message $\mu$ and a signature $\Sigma$.
\end{itemize}
\noindent
The correctness of \textsf{FSBS} is defined as follows.
For  any $ (pp, pk,sk_0) \leftarrow \textsf{Setup}(1^{n})$ and $ 
(\Sigma, \mathcal{V}) \leftarrow \textsf{Sign}(pp, pk,sk, t,\mu)$, 
the verification algorithm fails with a negligible probability or
\[
\Pr[\textsf{Verify}(pp,pk,t, \mu, \Sigma)=1]=1-\textsf{negl}(n).
\]
\subsection{Security of Forward-secure Blind Signatures} \label{sfs}
Two properties required for forward-secure blind signatures
are \textit{blindness} and  \textit{forward security}. 
Blindness ensures that it is impossible for the signer to learn 
any information about messages being signed. 
%===========================

\begin{definition}[Blindness] \label{bli}
$\mathsf{FSBS}$ is \textit{blind} if for any efficient algorithm  $\mathcal{S}^*$,
the advantage of $\mathcal{S}^*$ in the blindness game $\mathsf{Blind}_{\mathsf{FSBS}}^{\mathcal{S}^*}$ is negligible. 
That is  
\[
	\Adv^{\mathsf{Blind}}_{\mathsf{FSBS}}(\mathcal{S}^*):=\Pr[\mathsf{Blind}_{\mathsf{FSBS}}^{\mathcal{S}^*}
	\Rightarrow 1]-1/2 \leq \mathsf{negl}(n).
\]
$\mathsf{FSBS}$ is called \textit{perfectly blind} if $\Pr[\mathsf{Blind}_{\mathsf{FSBS}}^{\mathcal{S}^*}\Rightarrow 1]$ is exactly $1/2$. 
\end{definition}
\noindent
The blindness game $\textsf{Blind}_{\textsf{FSBS}}^{\mathcal{S}^*}$ consists of three phases defined below.

\begin{enumerate}
\item  \textbf{Initialization.} The adversary $\mathcal{S}^*$ chooses a security parameter $n$, 
	then obtains common parameters $pp$, a public key $pk$ and an initial secret key $sk_0$ using \textsf{Setup}$(1^{n})$. 
\item \textbf{Challenge.} $\mathcal{S}^*$ selects and gives  the challenger $\mathcal{C}$ two messages $\mu_0$ and  $\mu_1$. The challenger $\mathcal{C}$ flips a coin $b \in \{0, 1\}$ and initiates two signing interactions with  $\mathcal{S}^*$ on input $\mu_b$ and $\mu_{1-b}$ (not necessarily in two different time periods). The adversary $\mathcal{S}^*$ acts as the signer in these two interactions and finally attains two corresponding view/signature pairs $(\mathcal{V}_b,\Sigma_b)$ and $(\mathcal{V}_{1-b},\Sigma_{1-b})$.
	\item \textbf{Output.} The adversary $\mathcal{S}^*$ outputs $b' \in \{0,1\}$. It wins if $b'=b$.
	
\end{enumerate}

\noindent
Following \cite{DCK03}, we define forward-security 
as the \textit{forward-secure unforgeability}.  
In the $\textsf{FSUF}_{\textsf{FSBS}}^{\mathcal{U}^*}$ game, 
the forger $\mathcal{U}^*$ is a malicious user (adversary).
 
 	\begin{definition}[Forward-secure Unforgeability]   \label{on} 
 	$\mathsf{FSBS}$ is forward-secure unforgeable $(\mathsf{FSUF})$ if for any efficient algorithm  $\mathcal{U}^*$, the advantage of $\mathcal{U}^*$ in the forward-secure unforgeability game $\mathsf{FSUF}_{\mathsf{FSBS}}^{\mathcal{U}^*}$ is negligible. That is,
 	$$\Adv^{\mathsf{FSUF}}_{\mathsf{FSBS}}(\mathcal{U}^*):=\Pr[\mathsf{FSUF}_{\mathsf{FSBS}}^{\mathcal{U}^*}\Rightarrow 1] \leq \mathsf{negl}(n).$$ 
 	
 \end{definition}

 In our work, the forward-secure unforgeability game $\textsf{FSUF}_{\textsf{FSBS}}^{\mathcal{U}^*}$ in defined in the random oracle model. (We use hashing as an instantiation of random oracle.)
 We assume that, whenever the adversary wants to make a signing query, it always makes a random oracle query in advance.
\begin{enumerate}
\item \textbf{Setup.} The forger $\mathcal{U}^*$ gives a security parameter $n$ to the challenger $\mathcal{C}$. 
	The challenger $\mathcal{C}$ generates system parameters $pp$ and outputs the key pair $(pk, sk_0)$ 
	by calling \textsf{Setup}$(1^{n})$. 
	Then $\mathcal{C}$ sends $pp$ and  $pk$ to the forger $\mathcal{U}^*$. The key $sk_0$ is kept secret.
\item \textbf{Queries.} 
	At a time period $t$, the forger $\mathcal{U}^*$ can make a polynomially many random oracle queries 
	as well as a polynomially many signing queries in an adaptive manner.  
	In order to move to the next time period, the forger makes a key update query to get the secret key $sk_{t+1}$
	for the time period {$t$+1}.  
	Note that, once the forger makes a key update query, 
	i.e., it obtains the secret key $sk_{t+1}$, it cannot issue random oracle and signing queries 
	for past time intervals.
	Finally, the forger is allowed to make a single break-in query  
	at a time period $\overline{t}\leq T-1$, when it wants to stop the query phase.
	The time interval $\overline{t}$ is called the \textit{break-in time}.
	Once the forger makes the break-in query, it is not able to make further 
	random oracle (or hash) and signing queries.  
	Details of the challenger actions in response to the forger queries are given below.
	 \begin{itemize}
	\item For key update query $KQ(t)$: if $t <T-1$, then the challenger updates the secret key $sk_t$ to $sk_{t+1}$ 
	and updates $t$ to $t+1$. If $t=T-1$ then $sk_T$ is given as an empty string. 
	\item For each hash queries $HQ(t,\mu)$: the challenger has to reply with a random value. 
	\item For each signing query $SQ(t,\mu)$: the challenger must send a valid signature back to $\mathcal{U}^*$.
	\item For the break-in query $BQ(\overline{t})$ (note that the query is allowed once only): 
	the challenger must send  the secret key $sk_{\overline{t}}$ to the adversary and move the game to the output phase. 
	\end{itemize}
\item \textbf{Output.} $\mathcal{U}^*$ outputs at least one forgery  $ (\mu^*, t^*, \Sigma^*)$ at time period  $t^*$. He wins the game if $t^*< \overline{t}$, $SQ(t^*,\mu^*)$ has been never queried, and $(\mu^*, t^*, \Sigma^*)$ is valid.	
\end{enumerate}

%----section-our-construction
\section{Our Construction}

\subsection{Binary Tree Hierarchy for Time Periods} \label{binarytree}
Our design applies a binary-tree data structure. In the context of encryption, binary trees have been introduced by \cite{CHK03}.
For the lattice setting, they have been adapted by Cash et al. in  \cite{CHKP10}.
The tree structure is useful for constructing forward-secure public key encryption schemes  \cite{CHK03},
HIBE \cite{CHKP10} and recently for forward-secure group signature \cite{LNWX19}. 
We need time periods $t \in \{0, \cdots, 2^\ell-1 \}$ to be assigned to leaves of a binary tree of the depth $\ell$.
The tree leaves are arranged in increasing order from left to right -- see Figure \ref{fig2}.  
For a time period $t$, there is a unique path $t = (t_1, \cdots, t_\ell)$ from the root $\epsilon$ to the leaf,
where for each level $i\in [\ell]$, $t_i = 0$ if this is the left branch or $t_i = 1$ if this is the right branch.
Consequently, the $i$-th level node $w^{(i)}$ in the binary  tree can be described by
a unique binary bit string $w^{(i)} = (w_1, \cdots, w_i)$ that follows the path from the root to the node. 
This means that for the node $w^{(i)} =( w_1, \cdots, w_i)$, we can create
a corresponding matrix $W_t=[A_0\|A_1^{(w_1)}\| \cdots \|A^{(w_{i})}_{i}]$
(resp., $F_t=[A_0\|A_1^{(t_1)}\| \cdots \|A^{(t_{\ell})}_{\ell}]$), where $A_0$ and its ascociated  trapdoor $ T_{A_0}$ are generated by $\textsf{TrapGen}$ 
and $A_i^{(b)}$ are random matrices for all $i\in [\ell], b\in \{0,1\}$. 

Updating secret keys from time period $t$ to $t+1$ is done 
by the trapdoor delegation mechanism using \textsf{ExtBasis}. 
Each node $w^{(i)}=(w_1,\cdots, w_i)$ is associated with a secret key $T_{w^{(i)}}$, 
which can be computed from the initial secret key $sk_0=T_{A_0}$ by evaluating
\[
T_{w^{(i)}} \leftarrow \ExtBasis( A_{w^{(i)}}, T_{A_0}), \text{ where } A_{w^{(i)}} = \left[ A_0 \|A_1^{(w_1)} \|A_2^{(w_2)}\| \cdots \|A_i^{(w_i)} \right].
\]
$T_{w^{(i)}}$ is easily computed if a secret key $T_{w^{(k)}}$ for an ancestor $w^{(k)}$ of $w^{(i)}$ is known. 
Assume that the binary representation of $w^{(i)}$ is $w^{(i)} =( w_1, \cdots, w_k, w_{k+1}, \cdots, w_i)$,  
where $k<i$. 
Then 
$$T_{w^{(i)}} \leftarrow \ExtBasis( A_{w^{(i)}}, T_{w^{(k)}}), \text{ where } A_{w^{(i)}} = \left[ A_0 \|A_1^{(w_1)}\|\cdots \|A_2^{(w_k)}\| \cdots \|A_i^{(w_i)} \right].$$
Similarly, a secret key for a time period (i.e., a leaf) can be computed if we have any its ancestor's secret key.

\subsection{Description of the Proposed Signature}\label{forwardscheme}
Our lattice-based forward-secure blind signature (\textsf{FSBS}) consists of 
a setup algorithm \textsf{Setup}, a key update algorithm \textsf{KeyUp}, an interactive signing algorithm \textsf{Sign} 
and a verification algorithm \textsf{Verify}.  
They all are described below. Note that, we also use a commitment function \textsf{com}.

\begin{description}	 
\item \underline{\textsf{Setup}($1^n, 1^{\ell}$)}: 
	For a security parameter $n$ and a binary tree depth $\ell$, the algorithm runs through the following steps.
	\begin{itemize}
	\item Choose $q=poly(n)$ prime, $m=O(n\log q)$, $k$, $\kappa$, $\ell$, $\tau=2^\ell$, $\sigma$, $\sigma_1$, $\sigma_2$, $\sigma_3$ 
		(see Section \ref{paraset} for details).
	\item Let $\mathcal{M} = \{0,1\}^*$ be the message space of the scheme.
	\item Choose randomly a matrix  $K \xleftarrow{\$}\mathbb{Z}_q^{n\times k}$. Similarly, select 
		matrices $A_1^{(0)}, A_1^{(1)}, $ $A_2^{(0)}, A_2^{(1)}, \cdots,$ $ A_{\ell}^{(0)}, A_{\ell}^{(1)}$  from $\mathbb{Z}_q^{n\times m}$
		at random.
	\item Run $\TrapGen(q,n)$ to obtain a pair $(A_0, T_{A_0})$, 
		where $A_0 \in \ZZ_q^{n\times m}$ and $T_{A_0}\in\ZZ^{m\times m}$ are a matrix and its associated trapdoor.  
	\item Let $H: \{0,1\}^* \rightarrow \cal{R}_H$ be a collision-resistant and one-way hash function, where $\cal{R}_H:=\{\bf{e}' \in \{-1,0,1\}^k: \Vert \bf{e}' \Vert \leq \kappa \}$.
	\item Let $\mathsf{com}: \{0,1\}^* \times \{0,1\}^n \rightarrow \{0,1\}^n$ be a computationally binding and statistically hiding commitment function.
	\item Output $pp\leftarrow\{n,q, m,\ell,\tau, k, \kappa,\sigma, \sigma_1$, $\sigma_2, \sigma_3, \mathcal{M}, H, \textsf{com}\}$, $pk \leftarrow \{A_0, A_1^{(0)},$ $ A_1^{(1)},$ $ \cdots,$ $ A_\ell^{(0)}, A_\ell^{(1)}, K\}$, and $sk_{\epsilon} \leftarrow T_{A_0}$ as common parameters, public key and the initial secret key, respectively.
	\end{itemize} 
\item  \underline{$ \textsf{KeyUp}(pp,pk,sk_t,t)$}: 
	We need a key evolution mechanism (\textsf{KVM}) that ``forgets'' all secret keys 
	of internal nodes that can produce past keys.
	Additionally, we expect that \textsf{KVM} stores the smallest number of keys necessary for signature to work properly.
	The key evolution mechanism \textsf{KVM} works as follows.
   	\begin{itemize}
  	 \item 
   	For any leaf $t$, define the minimal cover $\textsf{Node}(t)$ to be 
	the smallest subset of nodes that contains an ancestor of all leaves in $\{t, \cdots , T-1\}$ but does not contain any
   	ancestor of any leaf in $\{0, \cdots , t-1\}$.   
   	For example, in Figure \ref{fig2}, $\textsf{Node}(0)=\{\epsilon\}$, $\textsf{Node}(1)=\{001, 01, 1\}$, 
	$\textsf{Node}(2)=\{01, 1\}$, $\textsf{Node}(3)=\{011, 1\}$ (i.e., two black circles in the tree), $\textsf{Node}(4)=\{1\}$,  
	$\textsf{Node}(5)=\{101, 11\}$, $\textsf{Node}(6)=\{ 11\}$, $\textsf{Node}(7)=\{111\}$.
	\item The secret key $sk_t$ at time period $t$ contains secret keys corresponding to all nodes 
	(including leaves) in  $\textsf{Node}(t)$. 
	For example, for the tree from Figure \ref{fig2},
	we have 
	$sk_0=sk_{\epsilon}=\{T_{A_0}\}$, $sk_1=\{T_{001}, T_{01}, T_{1} \}$, 
	where $T_{001}, T_{01}$, and $T_{1}$ are associated trapdoors for 
	$F_{001}=[A_0\|A_1^{(0)}\|A_{2}^{(0)}\|A_3^{(1)}]$, $F_{01}=[A_0\|A_1^{(0)}\|A_{2}^{(1)}]$ and $F_{1}=[A_0\|A_1^{(1)}]$, respectively.
	\item 
	To update $sk_t$ to $sk_{t+1}$, the signer determines the minimal cover $\textsf{Node}(t+1)$, 
	then derives keys for all nodes in $\textsf{Node}(t+1) \setminus \textsf{Node}(t)$ using the keys in $sk_t$ 
	as described in Section \ref{binarytree}.
	Finally the signer deletes all keys in $\textsf{Node}(t) \setminus \textsf{Node}(t+1)$. 
	For example, $sk_2=\{T_{01}, T_1\}$ (mentioned above), 
	since $\textsf{Node}(2) \setminus \textsf{Node}(1)=\{01,1\}$ and $\textsf{Node}(1) \setminus \textsf{Node}(2)=\{001\}$.
   	\end{itemize}
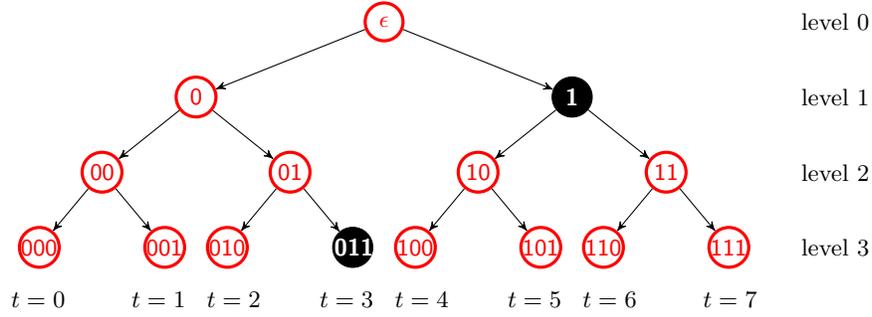
\begin{figure}
	\centering
\begin{tikzpicture}[->,>=stealth',level/.style={sibling distance = 5cm/#1,level distance = 1cm}] 
\node [arn_r] (O){$\epsilon$}
    child{ node [arn_r] (A) {0} 
            child{ node [arn_r] {00} 
            	child{ node [arn_r] (B) {000} }
                child{ node [arn_r] (C) {001}}
            }
            child{ node [arn_r] {01}
							child{ node [arn_r] (D) {010}}
							child{ node [arn_n] (E) {011}}
            }                            
    }
    child{ node [arn_n] {1}
            child{ node [arn_r] {10} 
							child{ node [arn_r] (F) {100}}
							child{ node [arn_r] (G) {101}}
            }
            child{ node [arn_r] {11}
							child{ node [arn_r] (H) {110}}
							child{ node [arn_r] (I) {111}}
            }
		}
; 
%\draw [->] (-6,-3.7) -- (6,-3.7);
\node at (-4.6,-3.7) {$t=0$};
\node at (-3,-3.7) {$t=1$};
\node at (-2,-3.7) {$t=2$};
\node at (-0.5,-3.7) {$t=3$};
\node at (0.5,-3.7) {$t=4$};
\node at (2,-3.7) {$t=5$};
\node at (3,-3.7) {$t=6$};
\node at (4.6,-3.7) {$t=7$};

\node at (6,0) {level $0$};
\node at (6,-1) {level $1$};
\node at (6,-2) {level $2$};
\node at (6,-3) {level $3$};
%\node at (D) [below ] {\( t=2 \)};
\end{tikzpicture}
	\caption{Binary tree of depth $\ell=3$, i.e., for $\tau=8$ time periods. The root is denoted by $\epsilon$. 
	For convenience, we name nodes by their binary representations}
	\label{fig2}
\end{figure}
\item \underline{$\textsf{Sign} (pp,pk, sk_t,t,\mu)$}:  
The signer interacts with the user in order 
to produce a signature for a message $\mu \in \cal{M}$ at time period $t$.
The interaction consists of five phases. Phases 1,3 and 5 are done by the signer.
Phases 2 and 4 -- by the user.
	\begin{itemize}
	\item \textit{Phase 1}: 
		The signer constructs the matrix $F_t=\left[A_0\| A_1^{(t_1)} \| \cdots \|A_\ell^{(t_\ell)}\right] \in \ZZ_q^{n \times (\ell+1) m}$
		for the time $t= (t_1, \cdots, t_\ell)$.
		Next it computes an ephemeral secret key $S_t$ using \textsf{SampleKey} described in 
		Theorem \ref{thm:Gauss}, where $F_t\cdot S_t=K$. 
		Note that $S_t$ can be computed at Phase 3 as well.
		The signer samples $\bf{r} \in \mathbb{Z} ^{(\ell+1)m}$    
		according to the distribution $\cal{D}^{(\ell+1)m}_{\sigma_2}$. It finally 
		computes and sends  $\mathbf{x}=F_t\mathbf{r} \in \ZZ_q^{n}$ to the user.
	\item \textit{Phase 2}: 
		Upon receiving $\textbf{x}$, the user samples blind factors 
		$\mathbf{a} \leftarrow \mathcal{D}_{\sigma_3}^{(\ell+1)m}$ and 
		$\mathbf{b} \leftarrow \mathcal{D}_{\sigma_1}^k$,  $\mathbf{d}' \xleftarrow{\$} \{0,1\}^n$. 
		It computes $\mathbf{u}=\mathbf{x}+F_t\mathbf{a} +K\mathbf{b}$ and 
		hashes it with $\bf{c}:=\mathsf{com}(\mu, \mathbf{d}') \in \{0,1\}^n$ 
		using the hash function $H$ to obtain a \textit{real challenge} $\mathbf{e}'$. 
		The rejection sampling technique is called
		 to get the \textit{blinded challenge} $\mathbf{e}$, which is sent back to the user. 
	\item \textit{Phase 3}: 
		The ephemeral secret key $S_t$ and $\textbf{r}$ are used to compute $\mathbf{z}=\mathbf{r}+S_t \mathbf{e}$. 
		In order to guarantee that no information of $S_t$ is leaked, 
		the rejection sampling is applied, which implies that
		the distribution of $\mathbf{z}$ and $\mathbf{r}$ are the same. 
		Finally, the \textit{blinded signature}  $\mathbf{z}$ is delivered to the user. 
	\item \textit{Phase 4}:  The user computes the unblinded signature $\mathbf{z}'=\mathbf{z}+\mathbf{a}$. 
		Again, the rejection sampling is called to make sure that $\mathbf{z}'$ and $\mathbf{z}$ 
		are independent of each other and $\mathbf{z}'$ is bounded in some desired domain.
		The user returns $(t,\mu, \Sigma=( \mathbf{d}', \mathbf{e}', \mathbf{z}'))$ as the \textit{final signature} 
		if $\Vert \mathbf{z}' \Vert \leq  \sigma_3\sqrt{(1+\ell)m} $ holds. Otherwise, he outputs ``$\bot$''. 
		The user is required to confirm validity of the final signature by sending $\textsf{result}$ 
	 	to the signer: $\textsf{result}:=\textsf{accept}$ means the final signature is good, 
		while $ \textsf{result}:=(\mathbf{a}, \mathbf{b}, \mathbf{e}', \bf{c})$) requires the user to
		restart the signing protocol.   	
	\item \textit{Phase 5}: Having obtained  
		$\textsf{result}$, the signer checks whether or not $\textsf{result}\neq \textsf{accept}$. 
		If not, it returns the \textit{view}  $\mathcal{V}=(t,\mathbf{r},\mathbf{e}, \mathbf{z})$.
		Otherwise it makes some check-up operations before restarting the signing algorithm. 
		The check-up allows the signer to detect an adversary who controls the user and tries to
		forge a signature.
	\end{itemize} 
Note that the rejection sampling in Phase 2 is not able restart the signing algorithm as it is used locally. 
In contrast, the rejection sampling in Phase 3 and Phase 4 can make the signing algorithm restart. 
The reader is referred  to Section \ref{sec5} for more details. 
Figure \ref{fig1} illustrates the signing algorithm.
		\begin{figure*}[h]
		\centering
		\medskip
		\smallskip
		\raisebox{\dimexpr 0.6\baselineskip-\height}% align tops

		\small\addtolength{\tabcolsep}{4pt}
		
		\begin{tabular}{|  l | l | }
		
			\hline
			&\\
			\underline{SIGNER $\mathcal{S}(pp,pk, sk_t,t)$:} & \underline{ USER $\mathcal{U}(pp, pk, t,\mu):$}\\
		
			&\\
			\textbf{\underline{Phase 1:}}& \textbf{\underline{Phase 2:}} \\
			01. $F_{t}:=\left[A_0\| A_1^{(t_1)} \| \cdots \|A_\ell^{(t_\ell)}\right] \in \ZZ_q^{n \times (\ell+1) m}$
			&05. $F_{t}:=\left[A_0\| A_1^{(t_1)} \| \cdots \|A_\ell^{(t_\ell)}\right]$ \\
			02. $S_t \in \mathbb{Z} ^{(\ell+1)m\times k} \leftarrow \textsf{SampleKey}(F_t,T_{F_t}, \sigma, K)$&06. $\mathbf{a} \xleftarrow{\$} \cal{D}^{(\ell+1)m}_{\sigma_3}$, $\bf{b} \xleftarrow{\$} \cal{D}^{k}_{\sigma_1}$	\\
			\hspace{0.5cm}	(i.e., $F_t\cdot S_t=K \text{ (mod } q)$)&07. $\mathbf{d}' \xleftarrow{\$} \{0,1\}^{n}$, $\bf{c}:=\textsf{com}(\mu,\mathbf{d})$,	\\
		03. $\bf{r} \in \mathbb{Z} ^{(\ell+1)m}\xleftarrow{\$} \cal{D}^{(\ell+1)m}_{\sigma_2}$, $\mathbf{x}=F_t\mathbf{r} \in \ZZ_q^{n}$&	\hspace{0.4cm}$\mathbf{u}=F_t\mathbf{a}+\mathbf{x}+K\mathbf{b} \text{ (mod } q)$\\
		
		04. Send $\mathbf{x}$ to the user & 08. $\mathbf{e}'=H(\mathbf{u}, \bf{c}) \in \cal{R}^k_{H}$, $\mathbf{e}:=\mathbf{e}'+\mathbf{b}$\\
	\hspace{0.5cm}	[\textbf{Go to Phase 2}]	&09. Output $\mathbf{e}$ with probability 	\\
			
			\textbf{\underline{Phase 3:}} & \hspace{1.5cm} min$ \left\{ \frac{\mathcal{D}_{\sigma_1}^m(\mathbf{e})}{M_1 \cdot \mathcal{D}^m_{\sigma_1,\mathbf{e}'}(\mathbf{e})},1 \right\}$\\
			11. $\mathbf{z}=\mathbf{r}+S_t\mathbf{e}$ &  10. Send $\mathbf{e}$ back to the signer.\\
			12. Output $\mathbf{z}$ with probability & \hspace{0.5cm}  [\textbf{Go to Phase 3}]\\
			 \hspace{2cm} min$ \left\{ \frac{\mathcal{D}_{\sigma_2}^{(\ell+1)m}(\mathbf{z})}{M_2 \cdot \mathcal{D}^{(\ell+1)m}_{\sigma_2,S_t\mathbf{e}}(\mathbf{z})},1 \right\}$ &\textbf{\underline{Phase 4:}}\\
					13. Send $\mathbf{z}$ to the user & 14. $\mathbf{z}'=\mathbf{z}+\mathbf{a}$\\

		 \hspace{0.5cm}[\textbf{Go to Phase 4}]&15.  Output $\mathbf{z}'$ with probability \\
			\textbf{\underline{Phase 5:}}	&\hspace{1.5cm} min$ \left\{ \frac{\mathcal{D}_{\sigma_3}^{(\ell+1)m}(\mathbf{z}')}{M_3 \cdot \mathcal{D}^{(\ell+1)m}_{ \sigma_3,\mathbf{z}}(\mathbf{z}')},1 \right\}$\\
	18. \textbf{if} (\textsf{result} $\neq$ \textsf{accept}):	&i.e.,  \textbf{if} ($\Vert\mathbf{z}' \Vert < \sigma_3\sqrt{(\ell+1)m}):$   \\
		19. \hspace{0.5cm} Parse \textsf{result} $:= (\mathbf{a}, \mathbf{b}, \mathbf{e}', \bf{c})$	& \hspace{1cm}  \textsf{result} $:= $ \textsf{accept} \\

			20.			\hspace{0.5cm}
			$\mathbf{u}:=F_t\mathbf{a}+\mathbf{x}+ K\mathbf{b}$ (mod $q$)	& 
			\hspace{0.5cm}\textbf{else}:  \textsf{result} $:= (\mathbf{a}, \mathbf{b}, \mathbf{e}', \bf{c})$\\
		\hspace{1.1cm}$\widehat{\mathbf{u}}:=F_t\mathbf{a}+F_t\mathbf{z}- K\mathbf{e}'$ (mod $q$)&16. \textbf{Output:} ($t, \mu$, $ \Sigma=(\mathbf{d}',\mathbf{e}', \mathbf{z}'))$\\

			21.		\hspace{0.5cm} \textbf{if} ($\mathbf{e}-\mathbf{b}=\mathbf{e}'=H(\mathbf{u}, \bf{c})$ & \hspace{1cm} or $\bot$ when \textsf{result }$\neq$ \textsf{accept}  \\
			  	\hspace{1.5cm} and $\mathbf{e}'=H(\widehat{\mathbf{u}}, \bf{c})$ & 17. Send $\textsf{        result      }$ back to the signer. \\

			\hspace{1.5cm}  and 	$\Vert \mathbf{z}+\mathbf{a}\ \Vert \geq \sigma_3\sqrt{(\ell+1)m})$:  & \hspace{0.5cm} [\textbf{Go to Phase 5}] \\

				\hspace{2cm}   restart from Phase 1 &\\
		
			%	\hline
			%	\textbf{\underline{Phase 6:}}&& \\
			22.		\textbf{Output:} the view $\mathcal{V}=(t, \mathbf{r},\mathbf{e}, \mathbf{z})$& \\
			\hline
		\end{tabular}
		
		\medskip
		\caption{The signing algorithm \textsf{Sign}($pp, pk, sk_t,t,\mu$)} 
		\label{fig1}
	\end{figure*}
\item \underline{\textsf{Verify}($t, pk, \mu, \Sigma$)}:  
	The algorithm accepts 
	a signature $\Sigma$ on the message $\mu$ for the time period $t= (t_1, \cdots, t_\ell)$ and
	public key $pk$ as its input and performs the following steps:\\
	(i) parse $\Sigma=(\mathbf{d}',\mathbf{e}', \mathbf{z}' )$;\\
	(ii) form $F_{t}:=\left[A_0\| A_1^{(t_1)} \| \cdots \|A_\ell^{(t_\ell)} \right] \in \ZZ^{n \times (1+\ell) m};$\\
	(iii) compute $\widehat{\mathbf{e}}:=H(F_t\mathbf{z'}- K\mathbf{e}' \text{ mod } q, \textsf{com}(\mu,\mathbf{d}'))$;\\
	(iv) if $\|\mathbf{z'}\| \leq \sigma_3\sqrt{(1+\ell)m}$ and $\widehat{\mathbf{e}}=\mathbf{e}'$,
	  then output 1, otherwise return 0. 
\end{description}

\section{Correctness, Security  and Parameters for \textsf{FSBS}}  \label{sec5}
%\input{sec-security-analysis}

%====================================

\subsection{Correctness}  \label{correct}

\begin{theorem}[\textbf{Correctness}] \label{theo3} 
The correctness of $\mathsf{FSBS}$ scheme holds after at most $e^2$ restarts 
with probability not smaller than $1-2^{-100}$.
\end{theorem}

\begin{proof} 
Given $(t,\mu,\Sigma= (\mathbf{d}', \mathbf{e}', \mathbf{z}'))$
produced by \textsf{Sign}($pp,pk, sk_t,\mu$) -- see Figure \ref{fig1}.
It is east to show that 
$H(F_t\mathbf{z}' -K\mathbf{e}'\text{ (mod } q), \mathsf{com}(\mu, \mathbf{d}'))=\mathbf{e}'$. 
\iffalse
Indeed, using the fact $K=F_t\cdot S_t$, we have 
	\begin{align*} F_t\mathbf{z}' -K\mathbf{e}'\text{ (mod } q)&=F_t(\mathbf{z}+\mathbf{a}) -K(\mathbf{e}-\mathbf{b})\text{ (mod } q)\\
	&=F_t\mathbf{z}+F_t\mathbf{a} -K\mathbf{e}+K\mathbf{b}\text{ (mod } q)\\&=F_t(\mathbf{r}+S_t\mathbf{e})+F_t\mathbf{a} -K\mathbf{e}+K\mathbf{b}\text{ (mod } q)\\
	&=F_t\mathbf{r}+F_t\mathbf{a} +K\mathbf{b}\text{ (mod } q)\\
	&=\mathbf{x}+F_t\mathbf{a}+K\mathbf{b} \text{ (mod } q).
	\end{align*}	
	\fi 
Note that $\Vert \mathbf{z}'\vert \leq \sigma_3\sqrt{(1+\ell)m}$ with overwhelming probability 
by Statement 1 of Theorem \ref{thm:Gauss}. 
Remark \ref{rem2} implies that if $s=12\Vert  \mathbf{c}\Vert$, 
then 
$	\frac{\mathcal{D}_{s}^m(\mathbf{x})}{M \cdot \mathcal{D}^m_{s, \mathbf{c}}(\mathbf{x})} \leq \frac{e^{1+1/288}}{M}$
with probability at least $1-2^{-100}$.
The rejection sampling requires that $\mathcal{D}_{s}^m(\mathbf{x})/(M \cdot \mathcal{D}^m_{s,\mathbf{c}}(\mathbf{x}))\leq 1$, 
meaning that $M \geq e^{1+1/288}$. It is easy to see that $M\approx e^{1+1/288}$ is the best choice. 
Applying this observation to the rejection samplings in Phases 3 and 4, 
we  see that a valid signature can be successfully produced after at most $M_2\cdot M_3 \approx e^2$ repetitions.
\end{proof}

%===========================
\subsection{Blindness}

\begin{theorem}[Blindness] \label{blindness} 
Let  $\mathsf{com}$ be a statistically hiding commitment and $H$ be an one-way and collision-resistant hash function.  
Then, the proposed forward-secure blind signature $\mathsf{FSBS}$ is blind.
\end{theorem}

\begin{proof} 
In the blindness game $\textsf{Blind}_{\textsf{FSBS}}^{\mathcal{S}^*}$, 
the adversarial signer $\mathcal{S}^*$ gives  the challenger $\mathcal{C}$ two messages $\mu_0$ and  $\mu_1$. 
The challenger $\mathcal{C}$ chooses uniformly at random a bit $b \in \{0, 1\}$ and 
interacts with $\mathcal{S}^*$ in order to sign both messages $\mu_b$ and $\mu_{1-b}$. 
$\mathcal{C}$ acts as two users $\mathcal{U}_b:=\mathcal{U}(pp, pk,t, \mu_b)$ 
and $\mathcal{U}_{1-b}:=\mathcal{U}(pp, pk,t, \mu_{1-b})$. 
Finally, $\mathcal{S}^*$ gets
two pairs $(\mathcal{V}_b,\Sigma_b)$ and $(\mathcal{V}_{1-b},\Sigma_{1-b})$
that correspond to the users $\mathcal{U}_b$ and $\mathcal{U}_{1-b}$, respectively.
We argue that the knowledge of $(\mathcal{V}_b,\Sigma_b)$ and $(\mathcal{V}_{1-b},\Sigma_{1-b})$ 
is independent of the signed messages. 
In other words, $\mathcal{S}^*$  
cannot distinguish,  which user it is communicating with. 
In other words, it cannot guess $b$ with non-negligible probability.

Indeed, for $\mathcal{V}_b=(t, \mathbf{r}_b,\mathbf{e}_b, \mathbf{z}_b)$ and 
$\mathcal{V}_{1-b}=(t, \mathbf{r}_{1-b},\mathbf{e}_{1-b}, \mathbf{z}_{1-b})$, 
we need to consider the pair $(\mathbf{e}_b,\mathbf{e}_{1-b})$ only, 
since  $ \mathbf{z}_b$ and  $\mathbf{z}_{1-b}$ are produced by $\mathcal{S}^*$ itself. 
In Phase 2, the rejection sampling  makes sure  that the distribution of 
both $\mathbf{e}_b$ and $\mathbf{e}_{1-b}$ are the same, which is $\mathcal{D}_{\sigma_1}^k$. 
This means that  $\mathbf{e}_b$ and $\mathbf{e}_{1-b}$ are independent of the signed messages. 	
Consider $\Sigma_b=(\mathbf{d}'_b,\mathbf{e}'_b, \mathbf{z}'_b)$ and
$\Sigma_{1-b}=(\mathbf{d}'_{1-b},\mathbf{e}'_{1-b}, \mathbf{z}'_{1-b})$.
As Phase 4 uses the rejection sampling,  both $\mathbf{z}'_{b}$ and $\mathbf{z}'_{1-b}$ 
have the same distribution, which is $\mathcal{D}_{\sigma_3}^{(1+\ell)m}$. 
It means that $\mathcal{S}^*$ does not learn anything about the signed messages
from the knowledge of $(\mathbf{d}'_b$, $\mathbf{d}'_{1-b})$ and ($\mathbf{e}'_b , \mathbf{e}'_{1-b}$). 
This is true because 
the former pair are randomly chosen and the latter pair are hash values of the one-way and collision-resistant function $H$.
	
Finally, it is easy to see that restarts, which may happen in Phase 5,
do not increase advantage of $\mathcal{S}^*$ in the blindness game. 
In fact, a restart occurs if the user has sent 
$\textsf{result}:= (\mathbf{a}, \mathbf{b}, \mathbf{e}', \bf{c})$ to $\mathcal{S}^*$.  
The values $\mathbf{d}'$, $\mathbf{a}$ and $\mathbf{b}$ are freshly sampled by the user. 
 Additionally, as \textsf{com} is a statistically hiding commitment, knowing $\bf{c}$,
 $\mathcal{S}^*$ cannot tell apart $\mu_b$ from $\mu_{1-b}$.
\end{proof}

%===========================================================%
\subsection{Forward-secure Unforgeability} \label{fseu}
We recall the following lemma, which we use to support
our witness indistinguishability argument. 
\begin{lemma}[Adapted from {\cite[Lemma 5.2]{Lyu12}}] \label{lem4}
Given a matrix $\mathbf{F} \in \mathbb{Z}_q^{n \times (\ell+1)m}$, where $(\ell+1)m>64+n\log q/\log(2d+1)$ and 
$\mathbf{s} \xleftarrow{\$} \{-d, \cdots, 0, \cdots, d\}^{(\ell+1)m}$. 
Then  there exists another $\mathbf{s}' \xleftarrow{\$} \{-d, \cdots, 0, \cdots, d\}^{(\ell+1)m}$ 
such that $\mathbf{F}\mathbf{s}=\mathbf{F}\mathbf{s}' \text{(mod } q)$
with probability at least $1-2^{-100}$.
\end{lemma}
Note that Lemma \ref{lem8} also gives the same conclusion as Lemma \ref{lem4} 
but with the not so clear condition $d \gg q^{(\ell+1)m/n}$.
\begin{theorem}[Forward-secure Unforgeability] \label{theo1} 
Suppose that the commitment function  $\mathsf{com}$ used in $\mathsf{FSBS}$ is 
computationally binding and that there exists a forger $\mathcal{A}$, 
who can break the forward-secure unforgeablity of $\mathsf{FSBS}$. 
Then, one can construct a polynomial-time algorithm $\mathcal{B}$ 
that solves an $l_2$-$\mathsf{SIS}_{q,n,(1+2\ell)m, \beta}$ 
problem with $\beta=\max\{(2\sigma_3+2\sigma\sqrt{\kappa})\sqrt{(1+\ell)m}, (2\sigma_3+\sigma_2)\sqrt{(1+\ell)m}\}$. 
	\iffalse
	The success probability $p$ of $\mathcal{B}$  is  at least \begin{align*}
		\min \frac{1}{\tau}\left\{ \frac{1}{2}(1-\gamma)\left(1-\frac{1}{|\cal{R}_H|}\right)\left(\frac{\delta-\frac{1}{|\cal{R}_H|}}{q_H}-\frac{1}{|\cal{R}_H|}\right), \delta\left(1-\frac{1}{|\cal{R}_H|}\right) \right\}.
	\end{align*} 
	\fi
\end{theorem}

\begin{proof}
The reduction is as follows:
\begin{description}
\item  \textbf{Phase 0 (Instance).} Assume that $\mathcal{B}$ wants to solve an instance of the $\mathsf{SIS}_{q,n,(1+2\ell)m, \beta}$ problem 
\begin{equation}\label{key}
F\cdot \mathbf{v}=0 \mod q, \Vert  \mathbf{v} \Vert \leq \beta, F\in \mathbb{Z}_q^{n\times (1+2\ell)m},
\end{equation}
in which $F$ is parsed as 
$ F = \left[ A_0 \|U_1^{(0)} \|U_1^{(1)}\| \cdots \|U_\ell^{(0)}\|U_\ell^{(1)} \right]$
with $ A^{(k)}_0, U_i^{(b)} \in \mathbb{Z}_q^{n\times m}$ for  $\beta=\max\{(2\sigma_3+2\sigma\sqrt{\kappa})\sqrt{(1+\ell)m}, (2\sigma_3+\sigma_2)\sqrt{(1+\ell)m}\}$ and $b \in \{0,1\}$.

%\textbf{Game 0.} This is the original \textsf{sUF-CMA} game.

%\textbf{Game 1.} This game is similar to Game 0, except that 

\item  \textbf{Phase 1 (Guessing the target).} $\cal{B}$ guesses the target time period $t^*$ that 
$\cal{A}$ wants to attack by choosing randomly $t^*=(t^*_1,\cdots, t^*_\ell) \xleftarrow{\$} \{0, \cdots, \tau-1\}$. 
The success probability of guessing $t^*$ is $1/\tau$.

\item  \textbf{Phase 2 (Initialize).} $\cal{B}$ sets common parameters $pp$ as in the \textsf{Setup} algorithm. 
However, $\cal{B}$ sets the public key $pk$ according to the following steps.
\begin{itemize}
	\item  For $i \in [\ell]$, $\cal{B}$ sets $A_i^{(t^*_i)}=U_i^{(t^*_i)}$. For each bit $b\in \{0,1\}$ such that $ b \neq t^*_i$, $\cal{B}$ invokes \textsf{TrapGen} to generate $A_i^{(b)}$ together with a short basis $T_{A_i^{(b)}}$ of $\Lambda_q^{\bot}(A_i^{(b)})$.

\item $\cal{B}$ samples $S^* \leftarrow \cal{D}_\sigma^{(1+\ell)m\times k}$ 
and sets $K:=F_{t^*}\cdot S^*$, where $F_{t^*}= \left[ A_0 \|A_1^{(t^*_1)}\| \cdots \|A_\ell^{(t^*_\ell)} \right] \in \mathbb{Z}_q^{n\times (1+\ell)m}$.  
Let $d:=\sigma\sqrt{(1+\ell)m}$. 
Then $\sigma$ should be chosen sufficiently large to satisfy Lemma \ref{uniform} (i.e., $\sigma \geq \omega(\sqrt{\log ((1+\ell)m)})$), 
Lemma \ref{lem8} (i.e., $d \gg q^{(1+\ell)m/n}$) and Lemma \ref{lem4} (i.e., $(1+\ell)m>64+n\log q/\log(2d+1)$).
Statement 1 of Theorem \ref{thm:Gauss} gurantees that 
$\|S^*\| \leq d$ with overwhelming probability. 
According Lemma 2, $K$ is statistically close to uniform. %  As we will see, $S^*$ will be exploited in answering signing queries at the time $t^*$.

\item Finally, $\cal{B}$  sends $pp$, and $pk \leftarrow \{ A_0, A_1^{(0)}, A_1^{(1)},  \cdots, A_{\ell}^{(0)}, A_{\ell}^{(1)}, K\}$ to $\cal{A}$ as the common parameters and the public key, while keeping $T_{A_i^{(b)}}$'s and $S^*$ secret.

\end{itemize}

$\mathcal{B}$ creates and maintains a list $\cal{L}_H$ consisting of random oracle queries 
$(\mathbf{u}, \bf{c})\xleftarrow{\$} \mathbb{Z}_q^{n} \times \{0,1\}^n$ and their corresponding hash value $ \mathbf{e}' \in \cal{R}_H$. In other words, $\cal{L}_H=\{(\mathbf{u}, \bf{c}, \bf{e}') \in  \mathbb{Z}_q^{n} \times \{0,1\}^n \times R_H: \bf{e}'=H(\mathbf{u}, \bf{c})\}$. In addition,  $\mathcal{B}$ also prepares the set of replies for $q_H$ hash queries $\mathcal{R}:=\{\mathbf{r}_1, \cdots, \mathbf{r}_{q_H} \} $, where each $\mathbf{r}_i \xleftarrow{\$} \cal{R}_H$. It then chooses a random tape $\rho$ and  runs $\mathcal{A}$ on $(pp, pk, \rho)$ in a black-box manner.

\item  \textbf{Phase 3 (Queries).} 
$\cal{B}$  plays the role of signer and interacts with  $\cal{A}$.
$\cal{B}$ responds to $\cal{A}$ queries as follows:
\begin{itemize}
\item \textit{Key update queries $KQ(t),t= (t_1, \cdots, t_\ell)$:}  If $t\leq t^*$, $\cal{B}$ aborts the query. 
Otherwise, let $k\leq \ell$ be the minimum index such that $t_k \neq t^*_k$. 
Then, the adversary $\cal{B}$ first uses the trapdoor $T_{A_k^{(t_k)}}$ to compute the key $T_{t_k}$ for the node $t_k$  
\[
T_{t_k} \leftarrow \ExtBasis(E\|A_k^{(t_k)}, T_{A_k^{(t_k)}}), \text{ where } E = \left[ A_0 \|A_1^{(t_1)} \| \cdots \|A_{k-1}^{(t_{k-1})} \right],
\]
 from which $\cal{B}$ computes all keys in $sk_t$ as in the real key update algorithm.  
\item \textit{Hash queries} $HQ(\bf{u},\bf{c})$: 	
Having received a hash query $(\mathbf{u}, \bf{c})$, 
$\mathcal{B}$ checks if the list $\cal{L}_H$ contains the query.
If $\mathcal{B}$ finds out that $(\mathbf{u}, \bf{c})$ is in $\cal{L}_H$ already, then $\mathcal{B}$ 
sends the corresponding hash value $\mathbf{e}'$ to the forger $\mathcal{A}$. 
Otherwise, $\mathcal{B}$ chooses the first unused $\mathbf{r}_i, i\in [q_H]$ from $\mathcal{R}$, 
takes $\mathbf{e}':=\mathbf{r}_i$ and stores the query-hash value pair $((\mathbf{u}, \bf{c}), \mathbf{e}')$ in $\cal{L}_H$.  
Finally, $\mathcal{B}$ sends $\mathbf{e}'$ to the forger $\mathcal{A}$ as the answer. 
\item \textit{Signing queries} 
$SQ(t,\mu)$: $\cal{B}$ constructs $F_{t}:=\left[A_0^{(k_0)}\| A_1^{(t_1)} \| \cdots \|A_\ell^{(t_\ell)} \right]$ 
and checks if $t\neq t^*$ or not. 
 If  $t\neq t^*$,  $\cal{B}$ computes 
 $T_{F_t} \leftarrow \ExtBasis(F_{t}, T_{A_k^{(t_k)}}),$ 
 and $S_t \leftarrow \mathsf{SampleKey}(F_{t}, T_{F_t}, $ $ \sigma, K)$, 
 where $k \leq \ell$ is the minimum index such that $t_k \neq t^*_k$.
 Note that  $F_t\cdot S_t=K$.
 Otherwise, if $t= t^*$, $\cal{B}$ simply assigns $S_{t^*} \leftarrow S^*$ since $F_{t^*}\cdot S^*=K$.

\item \textit{Break-in queries $BQ(t)$}: 
Once the adversary $\cal{A}$ makes a query $BQ(t)$, if $t\leq t^*$, then $\cal{B}$ aborts. 
Otherwise, i.e., $t>t^*$, $\cal{B}$ decides that the break-in time is $\overline{t} \leftarrow t$. $\cal{B}$
answers to $\cal{A}$ by sending the secret key $\mathsf{sk}_{\overline{t}}$ in the same way 
as replying to the key update queries since $\overline{t}= t>t^*$. 
\end{itemize}

\item \textbf{Phase 4 (Forge).} Eventually, $\cal{A}$ outputs a forgery $( t'_1, \mu^*_1, \Sigma_1^* )$. 
$\cal{B}$ checks if $t'_1=t^*$ or not. If not, then $\cal{B}$ aborts. 
Otherwise, $\cal{B}$ accepts the forgery. 
For the forgery $( t^*, \mu^*_1, \Sigma_1^* )$, we have: 
(i) $\Sigma_1^*=(\mathbf{d}'_1,\mathbf{e}'_1, \mathbf{z}'_1 )$; 
(ii)   $\mathbf{e}'_1:=H(F_{t^*}\mathbf{z'}_1- K\mathbf{e}'_1 \text{ mod } q, \textsf{com}(\mu^*_1, \textbf{d}'_1))$, 
where $F_{t^*}:=\left[A_0\| A_1^{(t^*_1)} \| \cdots \|A_\ell^{(t^*_\ell)} \right] \in \ZZ^{n \times (1+\ell) m}$; and 
(iii)  $\|\mathbf{z'}_1\| \le \sigma_3\sqrt{(1+\ell)m}$.\\
\end{description}

\noindent \textbf{Analysis.} 
We argue that the simulation of $\cal{B}$  is  statistically perfect. 
In other words, the forger $\cal{A}$  is not able to distinguish the 
simulator $\cal{B}$ from the real challenger in the \textsf{FSEU} game.
Indeed, the simulation proceeds as the real game except the following exceptions.
\vspace{-2mm}
\begin{enumerate}[label=(\roman*)]
\item Some matrices $A^{(b)}_i$ are not really random but is generated by \textsf{TrapGen}.
	However, Theorem \ref{thm:TrapGen} ensures that the distribution of $A^{(b)}_i$
	generated by \textsf{TrapGen} is close to uniform.
\item The matrix $K$ is not randomly chosen. It is obtained by sampling  $S^*$ from $\cal{D}_\sigma^{(1+\ell)m\times m}$
	and then assigning $K:=F\cdot S^*$. Lemma \ref{uniform} asserts that selection of $K$ is close to uniform. 
	Note that the sufficiently large choice of $\sigma$ does not affect (iii).
\item The matrix $S_{t^*}$ is equal to $S^*$, which is not computed using \textsf{SampleKey}.  
	The forger $\cal{A}$ does not know $S_t$ so consequently does not know $S^*$. 
	As $\bf{z}$ is generated (in Step 12) using the rejection sampling, 
	we always guarantee that $\bf{z}\leftarrow \cal{D}_{\sigma_2}^{(\ell+1)m}$ and $\bf{z}$ 
	is independent of $S_t$ and $S^*$. 
	Thus the view of $\cal{A}$ is independent of $S^*$.
\end{enumerate}
\vspace{-2mm}
Now, we show how to obtain the solution to the $l_2$-\textsf{SIS} problem given by Equation \eqref{key}. 
Let $i \in [q_H]$ be the target forking index, for which $\mathbf{e}'_1=\mathbf{r}_i$. 
$\mathcal{B}$ follows the rewinding strategy by keeping $\{\mathbf{r}_{1}, \cdots, \mathbf{r}_{i-1}\}$ 
and sampling new fresh answers $\{\mathbf{r}'_{i}, \cdots, \mathbf{r}'_{q_H}\} \xleftarrow{\$} \cal{R}_H$. Now,
$\mathcal{B}$ uses $\mathcal{R}':=\{\mathbf{r}_1, \cdots, \mathbf{r}_{i-1}, \mathbf{r}'_{i}, \cdots, \mathbf{r}'_{q_H} \}$
to answer to $\cal{A}$'s hash queries.
 
The forking lemma \cite[Lemma 4]{PS96} asserts that $\mathcal{A}$ outputs a new signature 
$( t'_2, \mu^*_2, \Sigma_2^* )$, where $\Sigma_2^*=(\mathbf{d}'_2,\mathbf{e}'_2, \mathbf{z}'_2 )$ such that $\bf{e}'_2=\bf{r}'_i$ 
using the same hash query as in the first run (i.e., the $i$-th hash query). 
Recall that $\gamma$ is the probability of a restart of \textsf{FSBS}.
As before, if $t'_2 \neq t^*$, then $\cal{B}$  aborts. 
If $\mathbf{e}'_2 = \mathbf{e}'_1$, 
$\mathcal{B}$ aborts and replays $\mathcal{A}(pp, pk, \rho')$ at most $q_H^{q_S}$ times 
using different random tapes $\rho'$ and different hash queries. 
If $\mathbf{e}'_2 \neq \mathbf{e}'_1$, 
then $\mathcal{B}$ returns 
\begin{equation}\label{key2}
 ((F_{t^*}\mathbf{z'}_1- K\mathbf{e}'_1 , \textsf{com}(\mu^*_1, \bf{d}'_1)), ( F_{t^*}\mathbf{z'}_2- K\mathbf{e}'_2 ,\textsf{com}(\mu^*_2, \bf{d}'_2)).
 \end{equation}
Since the pair in Equation \eqref{key2} are both coming from the same hash query 
and \textsf{com} is computationally binding,  
we have $\mu^*_2=\mu^*_1$, $\bf{d}_1'=\bf{d}_2'$ and  
\[
F_{t^*}\mathbf{z'}_1- K\mathbf{e}'_1 = F_{t^*}\mathbf{z'}_2- K\mathbf{e}'_2 \text{ (mod } q),
\]
or equivalently, 
\[
F_{t^*}(\mathbf{z'}_1-\mathbf{z'}_2- S^*(\mathbf{e}'_1-\mathbf{e}'_2 ))= \textbf{0} \text{ (mod } q).
\]
Set $\widehat{\bf{v}}:=\mathbf{z'}_1-\mathbf{z'}_2- S^*(\mathbf{e}'_1-\mathbf{e}'_2 )$.  
By Lemmas \ref{lem8} and \ref{lem4}, there is at least one secret key $S'$ such that 
$F_{t^*}S^*=F_{t^*}S' \text{ (mod } q)$, where $S^*$ and $S'$ have all the same columns 
except the $i$-th column.
The index $i$ shows the position, where $\mathbf{e}'_1[i] \neq \mathbf{e}'_2[i]$. 
If $\mathbf{z'}_1-\mathbf{z'}_2- S^*(\mathbf{e}'_1-\mathbf{e}'_2 )=\mathbf{0}$,
then we can choose $\widehat{v}:=\mathbf{z'}_1-\mathbf{z'}_2- S'(\mathbf{e}'_1-\mathbf{e}'_2 ) \neq \mathbf{0}$. 
Stress that the view of $\cal{A}$ is independent of both $S^*$ and $S'$.   
We have shown that $\widehat{\bf{v}} \neq \mathbf{0}$ and $F_{t^*}\cdot \widehat{\bf{v}} = \textbf{0} \text{ (mod } q)$.
 It is easy to see that $\Vert \widehat{\bf{v}}\Vert \leq 2(\sigma_3+\sigma\sqrt{\kappa})\sqrt{(1+\ell)m}$, 
 as   $\Vert S^*\Vert\leq \sigma\sqrt{(1+\ell)m}$, $\Vert \mathbf{z}'_{i}\Vert\leq \sigma_3\sqrt{(\ell+1)m}$, and $\Vert \mathbf{e}'_{i}\Vert\leq \sqrt{\kappa}$ for $i \in \{1,2\}$. 

In particular, we show that if $\mathcal{A}$ can produce a forgery by restarting
the signing interaction (with $\mathcal{B}$), 
then $\mathcal{B}$ is able to find a solution to the $l_2$-\textsf{SIS} problem  given by Equation \eqref{key}. 
Indeed, to restart the signing interaction,  
$\mathcal{A}$  delivers \textsf{result}$:=(\mathbf{a}, \mathbf{b}, \mathbf{e}', \mathbf{c})$ to $\mathcal{B}$.  
Now $\mathcal{B}$ with its view $\mathcal{V}=(t, \mathbf{r},\mathbf{e}, \mathbf{z})$, will check whether all
	\begin{align}
	\mathbf{e}-\mathbf{b}=\mathbf{e}'&=H(\mathbf{x}+F_{t^*}\mathbf{a}+K\mathbf{b} \text{ (mod } q), \bf{c})\label{eq3},\\
	\mathbf{e}'&=H(F_{t^*}\mathbf{a}+F_{t^*}\mathbf{z}-K\bf{e}' \text{ (mod } q), \bf{c}) \label{eq4},\\
	\Vert \mathbf{z}+\mathbf{a} \Vert &> \sigma_3 \sqrt{(1+\ell)m}. \label{eq5}
	\end{align}
hold or not. If all are satisfied, $\mathcal{B}$ restarts the interaction with $\mathcal{A}$. 
Let assume that afterwards $\mathcal{A}$ successfully produces a 
valid signature $ \widehat{\Sigma}=(\widehat{\mathbf{d}}', \widehat{\mathbf{e}}',\widehat{\mathbf{z}}')$. 
Let  $\widehat{\mathbf{b}} \in \mathcal{D}_{\sigma_1}^m$ be such that $\mathbf{e}=\widehat{\mathbf{e}}'+\widehat{\mathbf{b}}$. 
Then, the following relations have to hold
\begin{align}
	\mathbf{e}-\widehat{\mathbf{b}}=\widehat{\mathbf{e}}'&=H(\mathbf{x}+F_{t^*}\mathbf{a}+K \widehat{\mathbf{b}} \text{ (mod } q), \bf{c})\label{eq31},\\
	\widehat{\mathbf{e}}'&=H(F_{t^*}\widehat{\mathbf{z}}'-K\widehat{\mathbf{e}}'\text{ (mod } q), \mathsf{com}(\mu^*, \widehat{\mathbf{d}}')) \label{eq41},\\
	\Vert \widehat{\mathbf{z}}'\Vert &\leq \sigma_3\sqrt{(1+\ell)m}. \label{eq51}
\end{align}
Now, if $\widehat{\mathbf{e}}'\neq \mathbf{e}'$, then $\cal{B}$ aborts. 
Otherwise, Equations \eqref{eq4} and \eqref{eq41} give  
$F_{t^*}\mathbf{a}+F_{t^*}\mathbf{z} \text{ (mod } q)=F_{t^*}\widehat{\mathbf{z}}' \text{ (mod } q).$
Let $\widehat{\bf{v}}:=\mathbf{a}+\mathbf{z} -\widehat{\mathbf{z}}',$ then $\widehat{\bf{v}} \neq 0$.
This is true as otherwise 
$ \mathbf{a}+\mathbf{z}=\widehat{\mathbf{z}}' $, which implies that
$\Vert \mathbf{z}+\mathbf{a} \Vert \leq \eta \sigma_3 \sqrt{m} $ (by Equation \eqref{eq51}). 
This contradicts Equation \eqref{eq5}. 
Again, we have $F_{t^*}\cdot \widehat{\bf{v}} = \textbf{0} \text{ (mod } q)$,  
$\widehat{\bf{v}}\neq \mathbf{0}$ and 
$\Vert\widehat{\bf{v}} \Vert \leq \Vert\mathbf{a} \Vert+\Vert\mathbf{z} \Vert+\Vert\widehat{\mathbf{z}}' \Vert \leq (2\sigma_3+\sigma_2)\sqrt{(1+\ell)m}$.  
	
Note that $F_{t^*}= \left[ A_0 \|A_1^{(t^*_1)}\| \cdots \|A_\ell^{(t^*_\ell)} \right]=\left[ A_0 \|U_1^{(t^*_1)}\| \cdots \|U_\ell^{(t^*_\ell)} \right]$. 
We can get $F$ from $F_{t^*}$ by inserting into the gap between 
two sub-matrices in $F_{t^*}$ the remaining matrices $\{U_i^{(1-t^*_i)} \}_i$ at relevant positions. 
We insert zeros into the corresponding position of $\widehat{\bf{v}}$ 
to get the desired solution $\bf{v}$ to the problem given by Equation \eqref{key}. 
Obviously, $F\cdot \bf{v} = \textbf{0} \text{ (mod } q),$ and $\Vert \bf{v} \Vert=\Vert \widehat{\bf{v}} \Vert$. 

To summarise, we have shown that 
$\mathcal{B}$ can solve the $l_2$-\textsf{SIS}$_{q,n,(1+2\ell)m, \beta}$ 
problem, with $$\beta=\max\{ (2\sigma_3+2\sigma\sqrt{\kappa})\sqrt{(1+\ell)m}, (2\sigma_3+\sigma_2)\sqrt{(1+\ell)m}\}.$$
\end{proof}

\begin{table}[t]
	%   \begin{landscape}
	\begin{center}
		%\begin{adjustbox}{max width=\textwidth}
		%	\resizebox{\textwidth}{!}{
		\begin{tabular}{ |c|c| c|} 
			\hline
			\textbf{Parameters} & \textbf{Value}&\textbf{Usage}\\
			\hline
			\hline
			$n$&--&Security parameter\\
			\hline
			$\ell$&--&Binary tree depth\\
			\hline
			$\tau$&$2^\ell$&\#time points\\
			\hline
			$\beta$&\makecell{$\beta=\max\{ (2\sigma_3+2\sigma\sqrt{\kappa})\sqrt{(1+\ell)m},$\\ $ (2\sigma_3+\sigma_2)\sqrt{(1+\ell)m}\}$}&\multirow{2}{*}{\makecell{For $l_2$-\textsf{SIS}$_{q,n,(1+2\ell)m,\beta}$\\ to be hard, Theorem \ref{thm:SIS}}}\\
		
			$q$&$ q\geq \beta \cdot \omega(\sqrt{n\log n})$, prime&\\
			
			\hline
			$m$ &\makecell{$\max\{\frac{1}{1+\ell}\cdot(64+\frac{n\log q}{\log (2d+1)}),\lceil 6n\log q\rceil\}$,\\$d=\sigma \cdot \sqrt{(1+\ell)m}$}& Lemma \ref{lem4}, \textsf{TrapGen}\\
			\hline
		%	$K$&$m^{1+\epsilon}$, for any $\epsilon>0$& in \textsf{SampleKey}\\
			$\sigma$&$ \geq O(\sqrt{n\log q}) \cdot \omega(\sqrt{\log n})$& \textsf{SampleKey}, Theorem \ref{thm:Gauss}\\
			\hline
					
			%		$d$&$d=\sigma \cdot \sqrt{(1+\ell)m}$&\makecell{in simulation of the key $K$}\\

			$M_1,M_2,M_3$&$M_1=M_2=M_3=e^{1+1/288}$& \multirow{4}{*}{Rejection sampling}\\
			
			$\sigma_1$&$12\sqrt{\kappa}$&\\
		
			$\sigma_2$&$12\sigma\eta\sigma_1\sqrt{(1+\ell)mk}$&\\
			
			$\sigma_3$&$12\eta\sigma_2\sqrt{m}$&\\
			
			\hline
		%	signature size&$lm\log (12\sigma_3)+n+\kappa$ bits &\\
		%	secret key size&$lmk\log(2d+1)$ bits&\\
		%	public key size&$(lnm+nk)\log q$ bits&\\
				$k,\kappa$& $2^{\kappa}\cdot {{k}\choose{\kappa}}  \geq 2^{\gamma} $&\makecell{Min-entropy of the hash \\function $H$ at least $\gamma$}\\
						\hline
			\hline 
		\end{tabular} 
		%}
		%\end{adjustbox} 
	\end{center}
	\caption{Choosing parameters for the proposed \textsf{FSBS} scheme}
	\label{tab3}
\end{table} 
\vspace{-5mm}

\begin{remark}
In the proof for the forward-secure unforgeability, 
one may think of the method of programming hash values, instead of using the real signing interaction (with a modification in generating the matrix $S_t$ to compute $\mathbf{z}=\mathbf{r}+S_t\mathbf{e}$) in order to reply signing queries issued by $\mathcal{A}$. 
We argue that the programming method fails to simulate the perfect environment for the adversary $\mathcal{A}$. 
Assume that $\mathcal{B}$ does not want to compute $S_t$ in the way
we have done in our proof. 
Then, after replying to a hash query, 
say $(\mathbf{x}+F_t\mathbf{a}+K\mathbf{b} \text{ (mod } q), \textsf{com}(\mu, \textbf{d}'))$, 
by giving a hash value, say $\textbf{e}'$, 
$\mathcal{B}$ simply chooses $\textbf{z} \leftarrow \mathcal{D}_{\sigma_2}^{(\ell+1)m}$ 
and then sends $\textbf{z}$ to $\mathcal{A}$. In turn, $\mathcal{A}$ gives $\mathbf{e}:=\mathbf{e}'+\mathbf{b}$ to $\mathcal{B}$. 
After that $\mathcal{B}$ sets $H(F_t\mathbf{a}+F_t\mathbf{z}-K\mathbf{e}',\textsf{com}(\mu, \textbf{d}')):=\textbf{e}'$. 
However,  since the collision resistance of $H$, 
the relation
$F_t\mathbf{a}+F_t\mathbf{z}-K\mathbf{e}'=\mathbf{x}+F_t\mathbf{a}+K\mathbf{b} \text{ (mod } q)$ has to hold. 
Thus, $\mathcal{A}$  needs to check whether or not $F_t\mathbf{z} =\mathbf{x}+K\mathbf{e} \text{ (mod } q)$ 
to distinguish the simulated signing interaction from the real one. 
One may think that $\mathcal{B}$ can choose $\textbf{z} \leftarrow \mathcal{D}_{\sigma_2}^{(\ell+1)m}$  
such that $F_t\mathbf{z} =\mathbf{x}+K\mathbf{e} \text{ (mod } q)$ before sending $\textbf{z}$ to $\mathcal{A}$. 
However, without the knowledge of a trapdoor for $F_t$, the problem of choosing such a $\textbf{z}$ is  not easy.
\end{remark}

 \subsection{Choosing Parameters} \label{paraset}
First, we set $n$ as security parameter, $\ell$ as the highest depth of the binary tree representing time points,
$\tau=2^\ell$ as the number of time points.
For \textsf{TrapGen}, we need $m\geq \lceil 6n\log q\rceil$. 
For \textsf{SampleKey} (Theorem \ref{thm:Gauss}) to work, 
we need $\sigma \geq  O(\sqrt{n\log q}) \cdot \omega(\sqrt{\log n})$. 
Also, let $d:=\sigma\sqrt{(1+\ell)m}$ and we set $(\ell+1)m\geq 64+n \log q/\log (2d+1)$ via Lemma \ref{lem4}.
To make sure the min-entropy of $H$ is at least $\gamma$, we choose $k$ 
and $\kappa$ such that $2^{\kappa}\cdot {{k}\choose{\kappa}}  \geq 2^{\gamma} $.
Section \ref{correct} suggests setting $M_i:=e^{1+1/288}$ for all $i \in [3]$. 
We then set $\sigma_1= 12\Vert \mathbf{e}'\Vert=12\sqrt{\kappa}$, $\sigma_2= 12\Vert S^* \mathbf{e}\Vert=12\sigma\sigma_1\sqrt{(1+\ell)mk}$ 
and $\sigma_3= 12\Vert \mathbf{z}\Vert=12\eta\sigma_2\sqrt{(1+\ell)m}$ (via \textit{Remark} \ref{rem2}).
 For $l_2$-\textsf{SIS}$_{q,n,(1+\ell)m,\beta}$ to be hard by Theorem \ref{thm:SIS}, 
 we set  $m$ poly-bounded, $\beta=poly(n)$ and $q \geq \beta \cdot \omega(\sqrt{n\log n})$, where $\beta=\max\{ (2\sigma_3+2\sigma\sqrt{\kappa})\sqrt{(1+\ell)m}, (2\sigma_3+\sigma_2)\sqrt{(1+\ell)m}\}.$ The parameter setting is summarized in Table \ref{tab3}.

%----section-conclusion-future-works-----
\section{Conclusions and Future Works}  \label{conclusion}
%\input{sec-conclusion}

%====================================
In this paper, we propose, for the first time, a forward-secure blind signature based on the hardness of the SIS problem in lattices. 
Using the rejection sampling technique together with the trapdoor delegation and the binary tree structure 
for representing of time periods, the proposed signature is blind and forward secure.
Forward security is proven in the random oracle setting. 
Lattice-based forward-secure blind signatures in the standard model should be an interesting topic for future research.

\subsubsection{Acknowledgment.} 
We all would like to thank anonymous reviewers for their helpful comments.  
This work is partially supported by the Australian Research Council Discovery Project DP200100144 and Linkage Project LP190100984.   
Huy Quoc Le has been sponsored by a Data61 PhD Scholarship. 
Ha Thanh Nguyen Tran acknowledges the support of the Natural Sciences and Engineering Research Council of Canada 
(NSERC) (funding RGPIN-2019-04209 and DGECR-2019-00428). 
Josef Pieprzyk has been supported by the Australian ARC grant DP180102199 and Polish NCN grant 2018/31/B/ST6/03003.
%--------sec-references
%\input{sec-references}

\end{document}